\documentclass{llncs}

\usepackage{amsmath}
\usepackage{amsfonts}
\usepackage{amssymb}

\usepackage{enumerate}

\usepackage{url}

\newcommand{\keywords}[1]{\par\addvspace\baselineskip
\noindent\keywordname\enspace\ignorespaces#1}

\usepackage{ourmacros}

\begin{document}

\pagestyle{plain}

\mainmatter

\title{Branching Bisimilarity of Normed BPA Processes is in $\NEXPTIME$}

\author{Wojciech Czerwi\'{n}ski\inst{1}\thanks{W. Czerwinski
acknowledges a partial support by the Polish National Science Centre grant 2013/09/B/ST6/01575.}
\and Petr Jan\v{c}ar\inst{2}
%
%
}
\institute{University of Warsaw \and
FEI, Technical University of Ostrava \\
\email{wczerwin@mimuw.edu.pl, petr.jancar@vsb.cz}}

\maketitle

\begin{abstract}
\noindent
Branching bisimilarity on normed BPA processes
was recently shown to be decidable by Yuxi Fu (ICALP 2013) but
his proof has not provided 
any upper complexity bound. We present a simpler approach based on
relative prime decompositions that leads 
to a nondeterministic exponential-time algorithm; this is close to
the known exponential-time lower bound.
\end{abstract}

\section{Introduction}
Similarly as language equivalence in automata theory,
bisimulation equivalence (also called bisimilarity) is
a fundamental notion in theory of processes.
The decidability and complexity questions 
for bisimilarity on various models of infinite-state systems have been
explored in a long list of research papers.
(See~\cite{Srba:Roadmap:04} for an updated overview of a specific
area of process rewrite systems.)

One of the basic models is called Basic Process Algebra (BPA),
which can be related to context-free grammars in Greibach normal form.
The configurations are identified with sequences of variables
(nonterminals), and a configuration can change by performing an 
action (rather than reading a symbol) in which case its
leftmost variable is rewritten.
The seminal paper by Baeten, Bergstra and Klop~\cite{DBLP:journals/jacm/BaetenBK93} 
showed the decidability of bisimilarity for BPA configurations in the
normed case, where each variable can be stepwise
rewritten to the empty word;
this decidability result was later extended 
to the whole class BPA~\cite{DBLP:journals/iandc/ChristensenHS95}.
The exploration of complexity culminated by showing that the normed
case is, in fact, in $\PTIME$~\cite{DBLP:journals/tcs/HirshfeldJM96}
(see~\cite{CzerwinskiThesis} for the so far best known upper bound).
The complexity for the whole class BPA is known to lie between
$\EXPTIME$~\cite{DBLP:journals/ipl/Kiefer13} and $2$-$\EXPTIME$
(claimed in~\cite{DBLP:conf/mfcs/BurkartCS95} and explicitly proven 
in~\cite{DBLP:journals/corr/abs-1207-2479}).

In the presence of silent (unobservable) actions the 
problems become harder. 
The decidability question for weak bisimilarity of (even normed) BPA configurations
is a long-standing open problem; we only know  
$\EXPTIME$-hardness here,  already in the normed case~\cite{DBLP:journals/entcs/Mayr04}.
There is a similar long-standing open problem for Basic Parallel Processes, 
which is the parallel (or commutative) version of BPA. 
Positive results were recently achieved for a finer version of
weak bisimilarity, called \emph{branching bisimilarity}. (It was introduced
by van Glabbeek and Weijland~\cite{DBLP:journals/jacm/GlabbeekW96}
as the coarsest equivalence respecting branching time in some sense.)
It was shown that branching bisimilarity
is decidable on the normed 
Basic Parallel Processes~\cite{DBLP:journals/mst/CzerwinskiHL14}, and 
Yuxi Fu~\cite{DBLP:conf/icalp/Fu13} has shown the 
decidability for normed BPA configurations; the latter is the starting point
of our paper. 
We also note that the fresh paper~\cite{DBLP:conf/icalp/YinFHHT14}
shows that these decidability results cannot be essentially extended. 

Fu's result for branching bisimilarity on normed BPA is substantially
stronger than the previous results dealing with so called totally
normed
BPA~\cite{DBLP:conf/cav/Huttel91,DBLP:journals/iandc/CaucalHT95}; 
the proof uses an involved tableau framework 
(being inspired by~\cite{DBLP:conf/cav/Huttel91} and related works)
and does not provide any upper complexity bound.
Regarding the lower bound, Fu notes that the construction
used by Mayr~\cite{DBLP:journals/entcs/Mayr04} 
for weak bisimilarity can be easily adapted to yield 
$\EXPTIME$-hardness also for
branching bisimilarity on normed BPA.

An important novel ingredient of the decidability proof by
Yuxi Fu can be called the \emph{class-change norm} 
(corresponding to
the
branching norm in~\cite{DBLP:conf/icalp/Fu13}): 
while the standard norm counts
all the steps when a configuration is reduced to the empty one,
the class-change norm only counts the steps 
that change the current equivalence-class. 
It is not clear how to
compute this norm directly but equivalent configurations $\alpha\sim\beta$
must agree on this norm.
Another useful fact (also observed by Fu) is that the relation
of $\alpha\gamma$ and $\beta\gamma$ (either
$\alpha\gamma\sim\beta\gamma$ or $\alpha\gamma\not\sim\beta\gamma$) is
determined solely 
by the redundant variables w.r.t. $\gamma$, i.e. by those $X$
for which $X\gamma\sim \gamma$, independently of the string $\gamma$
itself.

\textbf{Our contribution} is based on introducing 
the \emph{relative prime decomposition} of configurations; unlike the prime decomposition of variables as,
e.g., in the case of bisimilarity of normed 
BPA~\cite{DBLP:journals/tcs/HirshfeldJM96}, we introduce a decomposition
related to each possible set $R$ of redundant variables (representing
respective suffixes $\gamma$). The (relative) equivalence
$\alpha\sim_R\beta$ can be then
replaced with the equality of prime decompositions
$\PD_R(\alpha)=\PD_R(\beta)$.
We suggest a nondeterministic exponential-time algorithm that guesses
the appropriate decompositions and then verifies their correctness in
the sense that the equality of the decompositions w.r.t. the guess is
indeed a branching bisimulation.
We thus place the branching bisimilarity of normed BPA configurations in
$\NEXPTIME$.

Fu~\cite{DBLP:conf/icalp/Fu13} has also shown that the respective
``regularity problem'' (given a normed BPA configuration, is it branching
bisimilar to some unspecified finite-state process?) is decidable. Our
approach places this problem in $\NEXPTIME$ as well.

\emph{Remark.}
It seems natural to look for a deterministic exponential-time
algorithm that would (deterministically) compute the decompositions, 
e.g., by proceeding via a certain series of decreasing overapproximations. 
Nevertheless, this question is left open here. 

\emph{Structure of the paper.}
In Section~\ref{sec:prelim} we define the used notions and state the
result. Section~\ref{sec:udp} deals with the class-change norm and
related observations that essentially already appeared 
in~\cite{DBLP:conf/icalp/Fu13}. 
Section~\ref{sec:relative-prime-decomposition} then introduces the relative prime
decompositions of configurations and proves their uniqueness.
Section~\ref{sec:bases} defines the branching bisimilarity bases and shows
how their consistency can be checked.
Finally Section~\ref{sec:algor} contains algorithms for the bisimilarity and regularity problems.

\section{Preliminaries, and statements of results}\label{sec:prelim}

We first recall a general definition of branching bisimilarity, 
which is then applied to BPA configurations.
Throughout the text we add some remarks (in italic) related to weak
bisimilarity; they are not needed for understanding the presented result.

By $\A^*$ we denote the set of finite sequences of elements of the
set $\A$. By $\eps$ we denote the empty sequence, and by $|w|$ the
length of $w\in\A^*$.
We put $\N=\{0,1,2,\dots\}$.

\subsubsection*{Labelled transition systems.}

 A \emph{labelled transition system}, an \emph{LTS} for short,
 is a tuple
$$\calL=(\calS,\act, (\trans{a})_{a\in\act})$$
 where $\calS$ is a set of states (at most countable in our case),
 $\act$ is a set of \emph{actions} (finite in our case), and
 $\trans{a}\subseteq \calS\times\calS$ is a set of \emph{transitions
 labelled with} $a$.
 We reserve the symbol
 \begin{center} 
 $\tau$ for  the (only) \emph{silent action};
\end{center}
 the  \emph{visible actions} are the elements of
$\actvis=\act\smallsetminus\{\tau\}$.
(If $\tau\in\act$, then  $\act=\actvis\cup\{\tau\}$, otherwise
$\act=\actvis$.)

We write $s\trans{a}t$ rather than $(s,t)\in\trans{a}$ (for $a\in\act$),
 and we define
 $s\trans{w}t$ for $w\in\act^*$ inductively:
 $s\trans{\eps}s$; if $s\trans{a}s'$ and
 $s'\trans{u}t$, then $s\trans{au}t$. By $s\trans{w}t$ we sometimes
also refer to a concrete respective path from $s$ to $t$ in $\calL$.

\subsubsection*{Branching bisimilarity.}
Given an LTS $\calL=(\calS,\act, (\trans{a})_{a\in\act})$,
a \emph{symmetric} relation $\calB\subseteq \calS\times\calS$ is a \emph{branching
 bisimulation} if 
 for any $(s,t)\in\calB$, $a\in\act$, and $s'\in\calS$ such that 
  $s\trans{a}s'$
 we have:
 \begin{itemize}
	 \item		 
$a=\tau$ and $(s',t)\in \calB$, or
\item
($a\in\actvis\cup\{\tau\}$ and)
 there is a sequence
 $t=t_0\trans{\tau}t_1\trans{\tau}\cdots\trans{\tau}t_k\trans{a}t'$
(for some $k\geq 0$) such that $(s',t')\in\calB$ and
$(s,t_i)\in\calB$ for all $i\in\{1,2,\dots,k\}$.
\end{itemize}
%
%
%
 By $s\sim t$, to be read as ``states $s,t$ are \emph{branching bisimilar}'',
 we denote that there is a branching bisimulation containing $(s,t)$.
We can easily verify the standard facts that $\sim$ is the union of all branching bisimulations,
and thus the maximal branching bisimulation, and that $\sim$ 
is an equivalence
 relation.
 
 \begin{remark} 
\emph{Weak bisimulations}, and \emph{weak bisimilarity}, are defined as above but
 we allow sequences with ``post'' $\tau$-transitions, like
 $t=t_0\trans{\tau}t_1\trans{\tau}\cdots\trans{\tau}t_k\trans{a}t'_0\trans{\tau}t'_1\trans{\tau}\cdots\trans{\tau}t'_\ell=t'$
 (for $\ell\geq 0$), and we only require that the final pair 
 $(s',t')$ belongs to $\calB$.
We denote the weak bisimilarity by $\approx$.
Any branching bisimulation is thus a weak bisimulation;
hence branching bisimilarity $\sim$ is finer than weak bisimilarity
$\approx$.
They coincide in the case with no silent action, in which case we
use the notion of (strong) bisimilarity.
In the system given by the following transitions we have $s_1\not\sim
s_2$ but $s_1\approx s_2$:
$s_1\trans{\tau}s_2$, $s_1\trans{a}s_5$, $s_2\trans{\tau}s_3$,  
$s_3\trans{a}s_5$,
$s_2\trans{a}s_4$, $s_4\trans{b}s_5$.
\end{remark}

\subsubsection*{Normed BPA systems.}

 A~\emph{BPA system} is given by a context-free grammar in Greibach
 normal form, with no starting variable (nonterminal). We denote it as
$$\calG=(\var,\act,\rules)$$
where
 $\var$ is a finite set of \emph{variables} (or nonterminals), $\act$
 is a finite set of \emph{actions} (or terminals), which can contain
 the \emph{silent action} $\tau$, and $\rules$ is a finite set of
 \emph{rules} of the form
 $A\trans{a}\alpha$ where $A\in\var$, $a\in\act$,
 $\alpha\in\var^*$. 

A BPA system $\calG=(\var,\act,\rules)$
has the associated LTS 
$$\calL_\calG=(\var^*,\act,(\trans{a})_{a\in\act})$$
where 
each rule 
 $A\trans{a}\alpha$ in $\rules$ induces the transitions
$A\beta\trans{a}\alpha\beta$ for all $\beta\in\var^*$.
The states of $\calL_\calG$, i.e. the strings of variables,
are called \emph{configurations} or \emph{processes}.

A \emph{variable} 
$A\in\var$
is \emph{normed}
if there is
$w\in\act^*$ such that $A\trans{w}\eps$. 
A \emph{BPA system}  $\calG=(\var,\act,\rules)$ is \emph{normed},
i.e. an  \emph{nBPA system}, 
if each $A\in\var$ is
normed.

\subsubsection*{Branching bisimilarity problem
for nBPA.}

By the \emph{branching bisimilarity problem for normed BPA} we mean the
decision problem specified as follows:
\begin{quote}
	\emph{Instance:} a normed BPA system
	$\calG=(\var,\act,\rules)$ and variables
$A,B\in\var$.
\\
	\emph{Question:} Is $A\sim B$ in $\calL_\calG$~?
\end{quote}
The variant with general configurations $\alpha,\beta\in\var^*$ in the
instances, asking whether  $\alpha\sim\beta$,
can be easily reduced to the above variant with variables $A,B$.

\subsubsection*{Semantic finitess (or regularity) problem.}
In our context, the \emph{regularity problem} is specified as follows:
\begin{quote}
	\emph{Instance:} a normed BPA system
	$\calG=(\var,\act,\rules)$ and 
$\alpha\in\var^*$.
\\
	\emph{Question:} Is $\alpha\sim s$ for a state 
	$s$ in some finite LTS\,?
\end{quote}
As usual, when comparing states in two different LTSs $\calL_1,\calL_2$, we
implicitly refer to the LTS 
arising as the disjoint union of $\calL_1$ and $\calL_2$.

\subsubsection*{Results.}
The next two theorems capture our main results.

\begin{theorem}\label{th:brbisnexptime}
The branching bisimilarity problem for normed BPA is in $\NEXPTIME$.
\end{theorem}	

\begin{theorem}\label{th:regulnexptime}
The regularity problem for normed BPA (w.r.t. branching bisimilarity)
is in $\NEXPTIME$.
\end{theorem}	

We prove the theorems in the following sections; in the rest 
of this section we recall some
facts about the standard norm, and we provide an example.

\subsubsection*{Standard norm.}
Given a
normed BPA system $\calG=(\var,\act,\rules)$, 
the norm $\norm{\alpha}$ of  
$\alpha\in\var^*$ 
is the length $|w|$ of a shortest
$w\in\act^*$ such that $\alpha\trans{w}\eps$.
(Note that the silent steps
$\trans{\tau}$ are also counted.)
%
A \emph{transition} $\alpha\trans{a}\beta$ is \emph{norm-reducing} if
$\norm{\alpha}>\norm{\beta}$, in which case
$\norm{\beta}=\norm{\alpha}{-}1$,
in fact.

The facts captured by the next proposition are standard and easy; they
also entail that we can check in polynomial time whether a BPA system
is normed.

\begin{proposition}
\hfill\\	
(1) $\norm{\eps}=0$. 
\\
(2) $\norm{\alpha\beta}= \norm{\alpha}+\norm{\beta}$.
\\
(3) $\norm{A}=1+\norm{\alpha}$ for a (norm-reducing) rule
$A\trans{a}\alpha$.
\\
(4) There is a polynomial-time algorithm (based on dynamic
programming) that  computes $\norm{A}$ for each
$A\in\var$ (when given 
$\calG=(\var,\act,\rules)$).
\\
(5) The 
values $\norm{A}$ are exponentially bounded (in the size of
$\calG$).
\end{proposition}


\begin{example}\label{ex:running}
Let
$\var=\{S_1,S_2,S_3\} \cup
	\{M_R\mid \emptyset\neq R \subseteq \{1, 2, 3\}\}
	\cup \{A,B,C,D\}$ and
$\act=\{a_1, a_2, a_3\}\cup\{\tau\}$.
We write just $M_{13}$ instead of $M_{\{1,3\}}$, etc.
Let $\rules$ be the set containing the following rules:
\begin{quote}
$S_1 \trans{a_1} \eps, \ S_1 \trans{\tau} \eps,
\ S_2 \trans{a_2} \eps,  \ S_2 \trans{\tau} \eps,
\ S_3 \trans{a_3} \eps,  \ S_3 \trans{\tau} \eps,$
\end{quote}
\begin{quote}
$M_R \trans{a_i} M_R,
\ M_R \trans{\tau} \eps$
for all nonempty $R \subseteq \{1, 2, 3\}$ and $i \in
R$, 
\end{quote}
\begin{quote}
	$A \trans{\tau} S_1 M_3, 
\ B \trans{a_1} C, 
\ B \trans{\tau} M_2 M_3,
\ C \trans{a_1}  C,
\ C \trans{\tau} M_3 M_2.$
\end{quote}
Here 
$\norm{S_i} = 1, \norm{M_R} = 1, \norm{A} = \norm{B} = \norm{C} = 3$.
We can check that $M_{23} \not \sim M_3 M_2$, and
$S_2 M_{23} \sim M_{23} \sim M_3 M_{23}$, though
$\norm{M_{23}}\neq\norm{M_3M_{23}}$.
\end{example}

\section{Class-change norm, and relative equivalences}\label{sec:udp}

In this section we recall some standard facts, the special norm introduced by Y.
Fu (given by a slightly modified definition here), and some observations presented already in~\cite{DBLP:conf/icalp/Fu13}. 
The main notions introduced here are the \emph{class-change norm} 
and the \emph{relative equivalences} $\sim_R$ for subsets $R$ of
the set of variables.
We implicitly refer to a given normed BPA system
$\calG=(\var,\act,\rules)$, though some claims hold more
generally. We give all proofs, to be self-contained.

\subsubsection*{Congruence property, and silent variables.}

\begin{proposition}\label{prop:congruence}
If $\alpha\sim\beta$ and $\gamma\sim\delta$,
then $\alpha\gamma\sim\beta\delta$.
\end{proposition}	

\begin{proof}
It is trivial to check that 
$\{(\alpha\gamma,\alpha\delta)\mid \alpha,\gamma,\delta\in\var^*,
\gamma\sim\delta\}$ is a branching
bisimulation. Hence $\gamma\sim\delta$ implies 
$\alpha\gamma\sim\alpha\delta$.
Slightly more subtle is to check that 
$\{(\alpha\gamma,\beta\gamma)\mid \alpha,\beta,\gamma\in\var^*, \alpha\sim\beta\}$ 
is a branching bisimulation.
\qed
\end{proof}


The second point in the above proof depends on our normedness
assumption.
(E.g., for the unnormed variable  $A$ with the only rule
$A\trans{\tau}A$
we have $\eps\sim A$ but $\delta\not\sim A\delta$
if $\delta\not\sim\eps$, since $A\delta\sim\eps$.)

We say that $A \in \var$ is a \emph{silent variable} if $A\trans{w}\alpha$ implies 
$w\in \{\tau\}^*$ (i.e., we can never perform a visible
action  when starting from $A$);
let $\varsil$ consist of all silent variables. 
We observe that  $\eps\sim\alpha$ iff $\alpha\in(\varsil)^*$ (in our
normed case); hence $\alpha\in(\varsil)^*$ implies 
$\gamma\sim\alpha\gamma$ for any $\gamma\in\var^*$.

\medskip

\textbf{Convention} (\emph{on silent variables}).
Since the silent variables can be determined by a straightforward
polynomial algorithm and they can be removed from any
$\beta\in\var^*$ without changing its equivalence class, we further
assume that our normed BPA systems have no silent variables.


\subsubsection*{Class-change norm.}

Example~\ref{ex:running} showed that
we can have $\alpha\sim\beta$ though $\norm{\alpha}\neq\norm{\beta}$.
We now define a norm for which this cannot happen.


A transition $\alpha\trans{a}\beta$ is \emph{class-changing}
if $\alpha\not\sim\beta$. The \emph{cc-length} of a path
$\alpha\trans{w}\beta$ is the number of class-changing transitions in
the path. We note that
the cc-length of $\alpha\trans{u}\beta\trans{v}\gamma$ is
the sum of the cc-lengths of $\alpha\trans{u}\beta$
and $\beta\trans{v}\gamma$.
\begin{center}
	The \emph{class-change norm} (or the \emph{cc-norm})
	of $\alpha\in\var^*$, 
denoted $\bnorm{\alpha}$,
\end{center}
is the minimum of   
the cc-lengths of paths $\alpha\trans{w}\eps$.

Any shortest path $\alpha\trans{w}\eps$ such that 
its cc-length is equal to $\bnorm{\alpha}$ is called 
\begin{center}
a \emph{witness
path for} $\alpha$. 
\end{center}
We observe that 
if  $\alpha\trans{u}\beta\trans{v}\eps$ is a witness path for $\alpha$
then $\beta\trans{v}\eps$ is a witness path for $\beta$.

\begin{remark}
Consider the rules 
$A\trans{\tau} A'$,
$A'\trans{\tau}\varepsilon$,
$A\trans{a}\varepsilon$,
$B\trans{a}B$,
$B\trans{b}\varepsilon$.
Here
$AB \sim A'B \sim B$, and the only witness path for $AB$ is 
$AB\trans{a}B\trans{b}\varepsilon$, with the cc-length $1$; hence
$\bnorm{AB}=1$ (while $\norm{AB}=2$). The branching norm 
in~\cite{DBLP:conf/icalp/Fu13} would be also $1$, but witnessed
by  $AB\trans{\tau}A'B\trans{\tau}B\trans{b}\varepsilon$.
\end{remark}

A transition $\alpha\trans{a}\beta$ is 
\emph{class-change-norm-reducing}, \emph{ccn-reducing} for short,
if  $\bnorm{\alpha}> \bnorm{\beta}$, i.e., if  
$\bnorm{\beta} = \bnorm{\alpha} {-} 1$.

\begin{proposition}\label{prop:bnormnecessary}
\hfill\\
	(1)  $\bnorm{\alpha}\leq\norm{\alpha}$.
\\
(2) If $\alpha\sim\beta$, then $\bnorm{\alpha}=\bnorm{\beta}$.
\end{proposition}

\begin{proof}
(1) is trivial.
\\
(2): Informally speaking, any class-changing transition must be matched
by a response that finishes by a corresponding class-change,
and thus the claim is intuitively
clear. Formally we
suppose a counterexample 
\begin{center}
$\alpha\sim\beta$,
$\bnorm{\alpha}<\bnorm{\beta}$,
\end{center}
where 
$\alpha$ has the shortest possible witness path $\alpha\trans{w}\eps$.
We cannot have 
$\alpha=\eps$, since in this case $\alpha=\beta = \eps$ 
(recall that we have excluded silent variables), 
and $\bnorm{\alpha}=\bnorm{\beta}=0$.
Hence $\alpha=A \delta$,
and the path $\alpha\trans{w}\eps$ can be written as 
\begin{center}
$\alpha=A\delta\trans{a}\gamma\delta=\alpha'\trans{w'}\eps$
\end{center}
(for a rule $A \trans{a} \gamma$).
Obviously, $\alpha' = \gamma\delta$ 
has a shorter witness path than $\alpha$, and
$\bnorm{\alpha'}\leq\bnorm{\alpha} < \bnorm{\beta}$.
Our assumptions thus imply  
$\alpha'\not\sim\alpha$
(otherwise $\alpha'\sim \beta$ 
with $\bnorm{\alpha'} < \bnorm{\beta}$
would constitute a ``smaller''
counterexample);
hence 
$\bnorm{\alpha'}=\bnorm{\alpha}-1$.
Since $\alpha\sim\beta$,
there is a response to the transition
$\alpha\trans{a}\alpha'$ (i.e., to
$A\delta\trans{a}\gamma\delta$), namely a sequence
\begin{center}
$\beta=\beta_0\trans{\tau}\beta_1\trans{\tau}\cdots\trans{\tau}\beta_k\trans{a}\beta'$
\end{center}
(for some $k\geq 0$),
such that $\alpha'\sim\beta'$ and
$\alpha\sim\beta_i$ for all $i\in\{0,1,\dots,k\}$; thus 
the transitions in 
$\beta_0\trans{\tau}\beta_1\trans{\tau}\cdots\trans{\tau}\beta_k$
are not class-changing.
Hence 
$\bnorm{\beta}\leq \bnorm{\beta'}{+}1$, and thus 
$\bnorm{\beta'}\geq \bnorm{\beta}{-}1>\bnorm{\alpha}{-}1
= \bnorm{\alpha'}$;
this implies that $\alpha',\beta'$ 
(where $\alpha'\sim\beta'$ and $\bnorm{\alpha'} < \bnorm{\beta'}$)
constitute a smaller
counterexample.
\qed
\end{proof}

In Example~\ref{ex:running} we can check that
$\bnorm{S_1 M_{12}} = 1$
and
$\bnorm{M_{12} S_1} = 2$, which implies
$S_1 M_{12} \not \sim M_{12} S_1$.

\begin{remark}
The equality of class-change norms of any pair of branching bisimilar
configurations $\alpha\sim\beta$ can help us to 
restrict the potential
consistent responses
$\beta=\beta_0\trans{\tau}\beta_1\trans{\tau}\cdots\trans{\tau}\beta_k\trans{a}\beta'$
to a transition $\alpha\trans{a}\alpha'$. In particular, all
$\beta_i$ ($i\in\{0,1,\dots,k\}$) must have the same class-change
norm. 
This is one of the points, which does not hold for weak bisimilarity.
\end{remark}

\subsubsection*{Redundant variables characterize the suffixes.}

The fact of compositionality, i.e., the fact that $\sim$ is a congruence
(Prop.~\ref{prop:congruence}), naturally leads us to look for possible
decompositions, in particular for the ``prime decompositions'' of
configurations $\alpha$, as we do in 
Section~\ref{sec:relative-prime-decomposition}.
It will turn out that the decomposition of $\alpha$ might be different in
the ``suffix-context'' $\alpha\gamma$ than in the context
 $\alpha\delta$ when $\gamma\not\sim\delta$.
Nevertheless, the decompositions of $\alpha$ in these two contexts 
will turn out to be the same if
$\red(\gamma)=\red(\delta)$, where for any $\beta\in\var^*$
we put
\begin{center}
$\red(\beta)=\{X\in\var\mid X\beta\sim \beta\}$.
\end{center}
\label{sec:relative-prime-decomposition}
The variables in $\red(\beta)$ are called  
the \emph{redundant variables w.r.t.} $\beta$. We observe:

\begin{proposition}\label{prop:redundprefix}
We have $\alpha\beta\sim\beta$ if, and only if, 
$\alpha\in(\red(\beta))^*$.
\end{proposition}
\begin{proof}
The ``if''-direction is obvious.

For the ``only-if''-direction assume 
$\alpha_1X\alpha_2\beta\sim\beta$ where 
$\alpha_2\in(\red(\beta))^*$ and $X\not\in\red(\beta)$;
hence $\alpha_1X\alpha_2\beta\sim\alpha_1X\beta\sim\beta$
and $X\beta\not\sim\beta$. This implies that the cc-length of any path
$X\beta\trans{u}\beta$ is positive, and thus
$\bnorm{\alpha_1X\beta}\geq\bnorm{X\beta}>\bnorm{\beta}$,
which excludes
$\alpha_1X\beta\sim\beta$
(since Prop.~\ref{prop:bnormnecessary}(2) implies
$\bnorm{\alpha_1X\beta}=\bnorm{\beta}$).
\qed
\end{proof}

To prepare a way for formalizing 
the above claims on decompositions,
we first relativize the
equivalence $\sim$ and the cc-norm $\bnorm{.}$ w.r.t. the suffix-contexts.

For  any $\gamma\in\var^*$ we define the relation $\sim_{\gamma}$ and 
the norm $\bnorm{.}_\gamma$ as follows:
\begin{center}
$\alpha\sim_\gamma\beta$ $\defiff$ $\alpha\gamma\sim\beta\gamma$, 
and $\bnorm{\alpha}_\gamma=_{df}\bnorm{\alpha\gamma}-\bnorm{\gamma}$.
\end{center}
%


\emph{Remark.}
In Example~\ref{ex:running} we have
$\bnorm{C}_\eps = 3$, but $\bnorm{C}_{M_2} = 2$ and $\bnorm{C}_{M_{23}} = 1$.
We can also check that
$\red(S_i) = \emptyset$, $\red(M_R) = \{M_S \mid S \subseteq R\} \cup \{S_i \mid i \in R\}$,
$\red(A) = \emptyset$, $\red(B) = \red(C) = \{S_1, M_1\}$,
$\red(S_1 M_{12}) = \red(M_{12})$.

\smallskip

We note some simple facts:

\begin{proposition}\label{prop:relnorm}
\ \hfill\\
(1)
If $\gamma\sim\delta$, then $\sim_\gamma=\sim_\delta$,
$\bnorm{.}_\gamma=\bnorm{.}_\delta$, and $\red(\gamma)=\red(\delta)$.
\\
(2) If $\alpha\sim_\gamma\beta$, then
	$\bnorm{\alpha}_\gamma=\bnorm{\beta}_\gamma$.
\\
(3)
$\bnorm{\alpha\beta}_\gamma=\bnorm{\alpha}_{\beta\gamma}+\bnorm{\beta}_\gamma$.
\\
(4)
$\alpha\gamma\sim\gamma$ iff $\alpha\in(\red(\gamma))^*$
iff $\bnorm{\alpha}_\gamma=0$;
\\
in particular, 
 $X\in\red(\gamma)$ iff $\bnorm{X}_\gamma=0$.
\\
\end{proposition}

\begin{proof}
(1), (2), and (3) follow trivially from the definitions and the fact that 
$\alpha\sim_\gamma\beta$ implies
$\bnorm{\alpha\gamma}=\bnorm{\beta\gamma}$ 
(by Prop.~\ref{prop:bnormnecessary}(2)).
\\
(4) partly repeats Prop.~\ref{prop:redundprefix}, and otherwise it
follows easily (by using (1), (2), (3)); in particular, if 
$\bnorm{\alpha}_\gamma=\bnorm{\alpha\gamma}-\bnorm{\gamma}=0$, then
any witness $\alpha\gamma\trans{w}\eps$ for $\alpha\gamma$ can be
written  $\alpha\gamma\trans{u}\gamma\trans{v}\eps$ where 
 $\alpha\gamma\trans{u}\gamma$ contains no class-change transition,
 which implies $\alpha\gamma\sim\gamma$.
%
\qed
\end{proof}

The next lemma says that any suffix-context is fully characterized
by the set of respective redundant variables:

\begin{lemma}\label{prop:redgammadetermines}
If $\red(\gamma)=\red(\delta)$, then 
$\sim_\gamma=\sim_\delta$ and $\bnorm{.}_\gamma=\bnorm{.}_\delta$.
\end{lemma}

\begin{proof}
Suppose $\red(\gamma) = \red(\delta)$; we will verify
that the set
\begin{center}
$\B=\{(\alpha\delta,\beta\delta)\mid \alpha\gamma\sim\beta\gamma \} \,\,
\cup\, \sim$
\end{center}
is a branching bisimulation.
We note that $\B$ is symmetric, and 
consider 
\begin{center}
 a pair $(\mu,\nu)\in\B$ and 
a transition 
 $\mu \trans{a} \mu'$.
\end{center}
 If $\mu\sim\nu$, then either $a=\tau$ and $\mu'\sim\nu$, in which case
 $(\mu',\nu)\in\B$, or there is a response
$\nu = \nu_0 \trans{\tau} \nu_1 \trans{\tau} \cdots \trans{\tau} \nu_k
\trans{a} \nu'$
such that $(\mu',\nu')$ and
$(\mu,\nu_i)$ for all $i\in\{1,2,\dots,k\}$ are in
$\sim$, and thus in $\B$.

Now assume 
\begin{center}
$(\mu,\nu)=(\alpha\delta,\beta\delta)$ where
$\alpha\gamma\sim\beta\gamma$.
\end{center}
If $\alpha=\eps$, then $\gamma\sim\beta\gamma$,
which entails
$\beta\in(\red(\gamma))^*=(\red(\delta))^*$, and thus 
$\mu=\alpha\delta=\delta\sim \beta\delta=\nu$
(by Prop.~\ref{prop:relnorm});
similarly, $\beta=\eps$
also implies $\mu\sim\nu$. We thus assume that $\alpha\neq\eps$ and
$\beta\neq\eps$; hence 
\begin{equation}\label{eq:mustep}
\mu=\alpha\delta \trans{a} \alpha'\delta=\mu'
\textnormal{ for a transition }
\alpha\trans{a}\alpha'.
\end{equation}
Since 
$\alpha\gamma\sim \beta\gamma$, the transition
$\alpha\gamma\trans{a}\alpha'\gamma$ entails that 
either $a=\tau$ and $\alpha'\gamma\sim\beta\gamma$,
in which case $(\alpha'\delta,\beta\delta)=(\mu',\nu)\in\B$,
or
there is a corresponding response 
of one of the following forms:
\begin{enumerate}
	\item
$\beta \gamma = \beta_0 \gamma \trans{\tau} \beta_1 \gamma \trans{\tau}
\cdots \trans{\tau} \beta_k \gamma \trans{a}
\beta' \gamma$,
\item
$\beta \gamma = \beta_0 \gamma \trans{\tau} \beta_1 \gamma \trans{\tau}
\cdots \trans{\tau} \beta_k \gamma
\trans{\tau}\gamma=\gamma_0\trans{\tau}\gamma_1\trans{\tau}\cdots\trans{\tau}\gamma_\ell\trans{a}\gamma'$.
\end{enumerate}
In the case $1$ we have that $(\alpha'\gamma,\beta'\gamma)$
and 
$(\alpha\gamma,\beta_i\gamma)$ for all $i\in\{1,2,\dots,k\}$ are in
$\sim$; hence  $(\alpha'\delta,\beta'\delta)$
and 
$(\alpha\delta,\beta_i\delta)$ for all $i\in\{1,2,\dots,k\}$
are in $\B$, and thus $\nu=\beta \delta
= \beta_0 \delta \trans{\tau} \beta_1 \delta \trans{\tau}
\cdots \trans{\tau} \beta_k \delta \trans{a}
\beta' \delta$ is an appropriate response to 
$\mu\trans{a} \mu'$ (i.e., to
$\alpha\delta \trans{a} \alpha'\delta$ from~(\ref{eq:mustep})).

In the case $2$ we have $\alpha\gamma\sim \gamma$, and thus also
$\gamma\sim\beta\gamma$. Hence both $\alpha,\beta$ are from 
$(\red(\gamma))^*=(\red(\delta))^*$,
which entails 
$\mu=\alpha\delta\sim \delta\sim\beta\delta=\nu$.

Since $\B$ is a branching bisimulation, we derive that
$\sim_\gamma \, \subseteq \, \sim_\delta$, and 
by symmetry that $\sim_\gamma = \sim_\delta$.
This also entails that $\bnorm{.}_\gamma=\bnorm{.}_\delta$.
Indeed, we recall that $\bnorm{\alpha}_\gamma$ is the smallest cc-length
(counting only the class-changing transitions) of the paths
$\alpha\gamma=\alpha_0\gamma\trans{a_1}\alpha_1\gamma\trans{a_2}\cdots\trans{a_n}\alpha_n\gamma=\gamma$.
Since $\alpha_{i-1}\gamma\trans{a_i}\alpha_i\gamma$ is a
class-changing transition (i.e.,
$\alpha_{i-1}\not\sim_{\gamma}\alpha_i$) iff 
$\alpha_{i-1}\delta\trans{a_i}\alpha_i\delta$ is a class-changing transition, we
derive that  $\bnorm{\alpha}_\gamma=\bnorm{\alpha}_\delta$.
\qed
\end{proof}	

\begin{remark}
The fact that just $\red(\gamma)$ determines whether
$\alpha\gamma\sim\beta\gamma$ or 
$\alpha\gamma\not\sim\beta\gamma$ 
is not true in the case of weak bisimilarity,
as illustrated by the example below.
\end{remark}

\begin{example}
Consider the system defined by:
\[
\begin{array}{rclcrclcrcl}
A & \trans{\tau} & \eps & \hspace{0.3cm} &
B & \trans{\tau} & \eps & \hspace{0.3cm}  &
X & \trans{a} & CX \\
A & \trans{b} & \eps & \hspace{0.3cm}&
B & \trans{b} & \eps & \hspace{0.3cm} &
X & \trans{x} & \eps  \\
A & \trans{a} & C & \hspace{0.3cm} &
C & \trans{c} & \eps\\
\end{array}
\]
One can easily check that $\red(X) = \emptyset$, clearly also $\red(\eps) = \emptyset$; here we define
the set of redundant variables with respect to the weak bisimilarity.
However, $AX \approx BX$, while $A \not\approx B$.
To see that $A \not\approx B$ note that $A$ have a transition labelled by $a$, while $B$ does not have such.
Let us show now that indeed $AX \approx BX$. For every transition from $BX$ there is an identical one from $AX$.
The only transition of $AX$, which cannot be matched by identical one from $BX$ is
\[
AX \trans{a} CX.
\]
However the following response from $BX$ is completely fine
\[
BX \trans{\tau} X \trans{a} CX.
\]
Note that the same response would not be correct in the case of branching bisimilarity as $X \not \sim AX$.
\end{example}

\subsubsection*{Relativization with respect to $R\subseteq\var$.}

We say that a \emph{set} $R\subseteq\var$ is 
 \emph{suffix-generated} if  $R=\red(\gamma)$ for some
 $\gamma\in\var^*$; we often
 implicitly consider only suffix-generated 
 $R\subseteq\var$ in what follows.
For any (suffix-generated) set $R\subseteq\var$,
 Lemma~\ref{prop:redgammadetermines} 
 allows us to soundly define the \emph{$R$-equivalence} $\sim_R$ by 
 \begin{center}
 $\alpha\sim_R\beta$  
if $\alpha\gamma\sim\beta\gamma$ for some $\gamma$ where 
$R=\red(\gamma)$,
\end{center}
and the \emph{cc$_R$-norm} $\bnorm{.}_R$ by
 \begin{center}
 $\bnorm{\alpha}_R=\bnorm{\alpha}_\gamma$
for some $\gamma$ where 
$R=\red(\gamma)$.
 \end{center}
%
%
 By the \emph{cc$_R$-length} of a path
$\alpha\trans{u}\beta$ we mean the number of transitions in the path 
that change the
class of $\sim_R$. By 
\begin{center}
an \emph{$R$-witness path for} $\alpha$ 
\end{center}
we mean a shortest path $\alpha\trans{u}\eps$ whose cc$_R$-length 
is minimal, and thus equal to $\bnorm{\alpha}_R$.

We also relativize $\red$ (for redundant variables), putting
\begin{center}
$\red_R(\alpha)=\red(\alpha\gamma)$ for some $\gamma$ where 
$R=\red(\gamma)$.
\end{center}
Abusing notation, we also write $\red(\alpha,R)$ instead 
of  $\red_R(\alpha)$. We note that  
$\alpha\in R^*$ implies
$\red(\alpha,R)=R$.

%
%
The next proposition summarizes some consequences of the previous
facts. 
We note that $\sim_R$ are not congruences in general, the
respective properties (captured by (7)) are more subtle.

\begin{proposition}\label{prop:Rsummary}
\hfill\\
	(1) $\red(\eps) = \emptyset$; $\sim \ = \ \sim_{\red(\eps)}$;
$\bnorm{.} = \bnorm{.}_{\red(\eps)}$.
\\
(2)
$\bnorm{\alpha}_R\leq \norm{\alpha}$.
\\
(3) 
$\alpha\sim_R\beta$ implies
$\bnorm{\alpha}_R=\bnorm{\beta}_R$.
\\
(4)
$\bnorm{\alpha\beta}_R=\bnorm{\alpha}_{R'}+\bnorm{\beta}_R$
where $R'=\red(\beta,R)$.
\\
(5)
$\alpha\sim_R\eps$ iff $\alpha\in R^*$
iff $\bnorm{\alpha}_R=0$;
in particular, $\bnorm{X}_R=0$ iff $X\in R$.
\\
(6)
If $\alpha\sim_R\beta$, then 
$\red(\alpha,R)=\red(\beta,R)$.
\\
(7) 
Suppose  $\alpha\sim_R\beta$ 
and $R'=\red(\alpha,R)$ ($=\red(\beta,R)$).
\\
Then $\gamma\sim_{R'}\delta$ iff 
$\gamma\alpha\sim_R\delta\beta$ (for any $\gamma,\delta$).
\end{proposition}

\begin{proof}
Points (1)--(6) are routine consequences of the definitions and
previous facts; in particular,  
$\red(\eps) = \emptyset$
since we have excluded silent variables.

(7): Let  $\alpha\sim_R\beta$ 
and $R'=\red(\alpha,R)$. We fix $\mu$ such that 
 $\alpha\mu\sim\beta\mu$ and $\red(\mu)=R$; hence
$\red(\alpha\mu)=\red(\alpha,R)=R'$.

We thus get: 
$\gamma\sim_{R'}\delta$
iff $\gamma\alpha\mu\sim\delta\alpha\mu$
iff  $\gamma\alpha\mu\sim\delta\beta\mu$
iff $\gamma\alpha\sim_R\delta\beta$.
\qed
\end{proof}

\section{Relative prime decomposition}\label{sec:relative-prime-decomposition}



Recall that we implicitly assume a normed BPA system $\calG=(\var,\act,\calR)$.

\subsubsection*{Redundancy-free form, relative prime variables, relative prime form.}

We say that $\alpha\in\var^*$ is \emph{$R$-redundancy-free},
for $R\subseteq \var$, 
if we do not have $\alpha=\beta X\gamma$ where $X\in\red(\gamma,R)$.
In other words, $\eps$ is 
$R$-redundancy-free for any $R$, and $\alpha X$ is $R$-redundancy-free
if $X\not\in R$ and $\alpha$ is $R'$-redundancy-free
for $R'=\red(X,R)$. When saying \emph{redundancy-free}, we
mean $R$-redundancy-free for $R=\red(\eps)=\emptyset$.

We define the \emph{$R$-redundancy-free form} $\rff_R(\alpha)$ of
$\alpha\in\var^*$ inductively:
\begin{enumerate}[i)]
	\item
		$\rff_R(\eps)=\eps$; 
	\item
		if  $X\in R$, then $\rff_R(\beta X)=\rff_R(\beta)$;
	\item
		if  $X\not\in R$, then $\rff_R(\beta X)=\rff_{R'}(\beta)\,X$ where 
$R'=\red(X,R)$.
\end{enumerate}
We note that $\alpha\sim_R\rff_R(\alpha)$, and 
$\alpha$ is $R$-redundancy-free iff $\alpha=\rff_R(\alpha)$.

\smallskip

\emph{Remark.}
Technically we could omit the explicit definition  
of $\rff_R(\alpha)$. It serves us mainly as a demonstration that our
inductive definitions and proofs for configurations proceed in the 
``right-to-left'' (or ``bottom-up'') fashion.

\smallskip

A \emph{variable} $A\not\in R$ is \emph{$R$-decomposable}
if $A\sim_R \alpha$ where $\alpha$ is $R$-redundancy-free 
and $|\alpha|>1$; 
if $A\not\in R$ is not $R$-decomposable, then $A$ is
\emph{non-$R$-decomposable}.

\begin{proposition}\label{prop:redfree}
\hfill\\
	(1) If $\alpha$ is $R$-redundancy-free, then 
$\bnorm{\alpha}_R\geq |\alpha|$.
\\
(2)
If $A$ is $R$-decomposable, then 
$A\sim_R \beta B$ where $B$ is non-$R$-decomposable and
$\bnorm{A}_R>\bnorm{B}_R\geq 1$.
\end{proposition}	
\begin{proof}
(1) follows from the already observed fact that
$X\gamma\not\sim\gamma$ implies that the cc-length of any path
$X\gamma\trans{u}\gamma$ is positive. 

(2): If $A$ is $R$-decomposable, then by definition we have $A\sim_R \beta B$
where $\beta B$ is $R$-redundancy-free, $\beta\neq\eps$, and 
$\bnorm{A}_R=\bnorm{\beta}_{R'}+\bnorm{B}_R$ where $R'=\red(B,R)$
(which follows from the definitions and Prop.~\ref{prop:Rsummary}).
Since $B\not\in R$ and 
$\beta$ is $R'$-redundancy free, we have $\bnorm{A}_R>\bnorm{B}_R\geq
1$.
Similarly, if $B$ is $R$-decomposable, then 
$B\sim_R \delta C$ where $\bnorm{B}_R>\bnorm{C}_R\geq 1$, which
entails $A\sim_R \beta B\sim_R \beta\delta C$.
We could thus immediately take $\beta$ and $B$
so that $\bnorm{B}_R$ is minimal, in which case $B$ is
non-$R$-decomposable.
\qed
\end{proof}

\textbf{$R$-prime form.}
For any (suffix-generated) $R\subseteq\var$, the equivalence $\sim_R$ induces a partition
on the set of all non-$R$-decomposable variables
(from $\var\smallsetminus R$). In each class of
this partition (containing non-$R$-decomposable variables 
that are pairwise equivalent w.r.t. $\sim_R$) we choose
a variable and call it an \emph{$R$-prime}.

\smallskip

\textbf{Convention.}
We further assume that the $R$-primes for all (suffix-generated)
sets $R\subseteq \var$ have been fixed, in our assumed nBPA system
$\calG=(\var,\act,\calR)$, unless stated otherwise.

\smallskip

We say that $\alpha\in\var^*$
is in the \emph{$R$-prime form} if 
\begin{itemize}
	\item
		either $\alpha=\eps$, 
	\item
		or 
$\alpha = \beta X$ where $X$ is an $R$-prime and $\beta$ is in the
$R'$-prime form for
$R'=\red(X, R)$.
\end{itemize}
The next lemma shows that for any $\alpha$ and
$R$ there is
the unique $\alpha'$ in the $R$-prime form such that
$\alpha\sim_R\alpha'$.

\begin{lemma}\label{lem:rel-udp}
\hfill\\
(1) For any  $\alpha\in\var^*$ and
$R\subseteq\var$ there is $\alpha'$ in the $R$-prime form where
$\alpha\sim_R\alpha'$.
\\
(2)
If $\alpha \sim_R \beta$ and both $\alpha$ and $\beta$ are in the
$R$-prime form, then $\alpha = \beta$.
\end{lemma}

\begin{proof}
(1): We show the claim for 
$\alpha,R$ by induction on $\bnorm{\alpha}_R$.
If  $\bnorm{\alpha}_R=0$, i.e. $\alpha\in R^*$, then $\alpha\sim_R\eps$
(where $\eps$ is in the $R$-prime form).
We thus assume $\alpha=\beta X \gamma$ where $\gamma\in R^*$ and
$X\not\in R$; note that $\alpha\sim_R\beta X\gamma\sim_R\beta X$.
Recalling Prop.~\ref{prop:redfree}(2), we deduce that 
$X\sim_R \delta Y$ where 
$Y$ is non-$R$-decomposable (maybe $\delta=\eps$ and $X=Y$);
hence $\alpha\sim_R\beta X\sim_R \beta \delta Y\sim_R 
\beta \delta Z$ where $Z$ is the
(unique) $R$-prime such that $Y\sim_R Z$.
We  put $R'=\red(Z,R)$, and note that 
 $\bnorm{\alpha}_R=\bnorm{\beta\delta}_{R'}+\bnorm{Z}_R$;
 hence $\bnorm{\beta\delta}_{R'}<\bnorm{\alpha}_R$.
By the induction
hypothesis there is $\beta'$ in the $R'$-prime form 
such that $\beta'\sim_{R'}\beta\delta$; hence $\alpha\sim_R\beta'Z$
and $\beta' Z$ is in the $R$-prime form.

(2):
Suppose 
the claim is not true.
Then there are $\alpha \sim_R \beta$, both $\alpha$ and $\beta$ being in the
$R$-prime form, such that $\alpha=\alpha'X\gamma$ and
$\beta=\beta'Y\gamma$ where $X,Y$ are different $R'$-primes for
$R'=\red(\gamma,R)$; hence $\alpha'X\sim_{R'}\beta' Y$, and 
by symmetry we can assume $\bnorm{X}_{R'}\geq\bnorm{Y}_{R'}$.

More generally, if the claim is not true, then we deduce that
there are some $R,\alpha,\beta$, and two \emph{different} $R$-primes $X,Y$ where
\begin{center}
$\alpha X\sim_R\beta Y$
and
$\bnorm{X}_{R}\geq\bnorm{Y}_{R}$.
\end{center}
We take such a counterexample with a shortest possible $R$-witness
path for $\alpha X$.
(Note that we do not require that $\alpha X$ and $\beta Y$ are in the $R$-prime form.)

We note that $\alpha X\not\sim_R X$ (since otherwise
$X\sim_R \beta Y$, and thus
either
$\bnorm{X}_{R}>\bnorm{Y}_{R}$ and $X$ is $R$-decomposable,
or
$X\sim_R\beta Y\sim_{R}Y$, which  
is impossible by our definition of $R$-primes).
We thus also have $\alpha X\not\sim_R Y$
(since $\bnorm{\alpha X}_R>\bnorm{Y}_R$).

Hence the first step in the $R$-witness path for $\alpha X$ is 
$\alpha X\trans{a}\alpha'X$ (where $\alpha' X$ has a shorter
$R$-witness path than $\alpha X$).
Since $\alpha X\sim_R \beta Y$, and  $\alpha X\not\sim_R Y$,
this first step 
 must have a response
$\beta Y=\beta_0 Y\trans{\tau}\beta_1
Y\trans{\tau}\cdots\trans{\tau}\beta_k Y\trans{a}\beta'Y$ where
$\alpha'X\sim_{R}\beta'Y$
(and all $\beta_i$ are nonempty since $\alpha X\not\sim_R Y$);
this contradicts the choice of our
counterexample.
\qed
\end{proof}
Lemma~\ref{lem:rel-udp} allows us to soundly define 
 \begin{center}
 $\PD_R(\alpha)$, the \emph{prime decomposition of} $\alpha$ w.r.t.
(i.e., in the context of) $R$,
\end{center}
as the unique configuration in the $R$-prime form such that
$\alpha\sim_R\PD_R(\alpha)$.

We explicitly note the following consequences:

\begin{corollary}\label{cor:realprimedec}
\hfill\\
(1) For any $R\subseteq \var$ and $\alpha,\beta\in\var^*$ we have
$\alpha\sim_R\beta$ iff $\PD_R(\alpha)=\PD_R(\beta)$.
\\
(2) The mapping $\PD_R(\alpha)$ (with arguments $R,\alpha$) 
satisfies the following:
\begin{enumerate}[i)]
	\item
$\PD_R(\eps)=\eps$\,;		
\item
if $X\in R$, then $\PD_R(X)=\eps$\,;
\item
	if $X$ is an $R$-prime, then $\PD_R(X)= X$\,;
\item
if $X\not\in R$, 
then 
$\PD_R(X)=\PD_{R}(\beta Z)$ for the unique $R$-prime $Z$ and any
$\beta$ such that $X\sim_R \beta Z$\,;
\item
$\PD_R(\alpha X)=\PD_{R'}(\alpha)\, \PD_{R}(X)$ where 
 $R'=\red(\PD_R(X),R)$.
\end{enumerate}
\end{corollary}

\smallskip

\emph{Remark.}
In the Example~\ref{ex:running},
$M_{12}$ is non-$\emptyset$-decomposable,
while $A$ is $\emptyset$-decomposable since $A \sim_\emptyset S_1 M_3$.
We have $\PD_\emptyset(C) = M_1 M_3 M_2$, so $\PD_{\{M_2, S_2\}}(C) = M_1 M_3$ and
$\PD_{\{M_{23}, M_2, M_3, S_2, S_3\}}(C) = M_1$.
Therefore $\PD_\emptyset(S_2 C M_{23}) = S_2 M_1 M_{23}$.

We note that we can have, e.g., $\PD_R(X) = XY$,
which may seem a bit unnatural.
This is illustrated by the example given by the following rules:
\begin{quote}
$X \trans{a} Y, 
\ Y \trans{b} Y, 
\ Y \trans{\tau} \eps.$
\end{quote}
Here $Y$ is $\emptyset$-prime, with $\bnorm{Y}_\emptyset = 1$, but
$X$ is not $\emptyset$-prime, since $X\sim_\emptyset XY$ and 
$\bnorm{X}_\emptyset = 2$. 
We have $Y\in\red(Y,\emptyset)$
(since $Y\sim_\emptyset YY$), and $X$ is $\{Y\}$-prime, with 
$\bnorm{X}_{\{Y\}} = 1$.

\section{Consistent bases}\label{sec:bases}

In this section we show that $R$-primes and $R$-decompositions
of non-$R$-primes 
of exponentially bounded sizes
can be guessed, and the consistency of the guess can be verified.
The candidates for decompositions will be captured by 
so called \emph{bases}; these are defined 
 via a technical notion of
 \emph{pre-bases}.

We again assume a given normed BPA system $\calG=(\var,\act,\rules)$,
now with no information about $R$-primes.
We now use the symbol $\B$ for (pre)bases, rather than for
bisimulations.

\subsubsection*{Pre-bases.}

A \emph{pre-base} $\B$ is determined by its domain 
$\dom(\B)\subseteq 2^\var$, where $\emptyset\in\dom(\B)$,
and two disjoint sets 
$\B_{dec}$ (\emph{decompositions}) and  $\B_{prop}$
(\emph{propagations})
containing triples $(A,\alpha,R)$ where
$A\in\var$, $\alpha\in\var^*$, $R\subseteq \var$, for which
the following
conditions hold:
\begin{enumerate}
	\item
	For each $R\in\dom(\B)$, the set $\var\smallsetminus R$ is partitioned
	into the set of \emph{$(\B, R)$-primes} and the set of 
\emph{$(\B, R)$-non-primes}. (Thus for $R=\emptyset$ each $X\in\var$
is either a $(\B, R)$-prime or a $(\B, R)$-non-prime.)
Moreover:
	\begin{itemize}
		\item	For each $(\B, R)$-non-prime $A$ there is precisely one triple $(A,\alpha,R)$ in
			$\B_{dec}$, where $\alpha=\beta B$ for a
			$(\B,R)$-prime $B$ (and $\beta$ can be empty).
		\item For each $(\B, R)$-prime $B$ there might be some triples $(B,XB,R)$ 
			in $\B_{prop}$ (where $X\in\var$). 
			We put 
			\begin{equation}\label{eq:propag}
\red^\B(B,R)=\{X\mid (B,XB,R)\in\B_{prop}\}.
\end{equation}
	\end{itemize}
	\item $\B_{dec}$ and $\B_{prop}$ do not contain any other triples than those mentioned above.	
	\item $\dom(\B)$ is equal to the least set that contains $\emptyset$ and is closed under propagation,
		which means that if $R'=\red^\B(B,R)$ for some $R\in\dom(\B)$ and some $(\B, R)$-prime $B$,
		then $R'\in\dom(\B)$.  
\end{enumerate}
Informally,  a triple $(A, \alpha, R) \in \B_{dec}$
can be viewed as a statement 
that $A$ is not an $R$-prime and that 
$\PD_R(A) = \alpha$ (which includes the case $\alpha=B$ for a
guessed $R$-prime $B$).
A triple $(A, XA, R) \in \B_{prop}$ 
can be viewed as a statement that 
$X\in\red(A,R)$
(and that $A$ is an $R$-prime).  

\textbf{$(\B,R)$-prime-decompositions.}
For a pre-base $\B$ and $R\in\dom(\B)$ we define 
the \emph{$(\B,R)$-prime-decomposition form} $\PD_R^\B(\gamma)$ of  
strings  $\gamma\in\var^*$ 
by the following (inductive) definition: 
\begin{enumerate}[i.]
	\item	$\PD_R^\B(\eps)=\eps$; 
	\item
		if $A\in R$, then 	$\PD_R^\B(\beta A)=\PD_R^\B(\beta)$; 
	\item
if $A$ is a $(\B, R)$-prime, then
$\PD_R^\B(\beta A)= \PD_{R'}^\B(\beta)\,A$ for  $R'=\red^\B(A,R)$;
\item 
if $A$ is a $(\B, R)$-non-prime,
then
		$\PD_R^\B(\beta A)= \PD_{R}^\B(\beta\alpha)$
		where $(A,\alpha,R)\in\B_{dec}$.
\end{enumerate}
Our definition of pre-bases has not excluded that 
 $\PD_R^\B(\gamma)$ could be infinite. 
 Nevertheless, our intention is
 that a triple $(A,\alpha,R)\in\B_{dec}$ captures 
 the guess $\PD_R(A)=\alpha$, which entails
 $|\alpha|\leq\bnorm{\alpha}_R=\bnorm{A}_R\leq \norm{A}$.
 We thus impose further (syntactic) conditions:

 \subsubsection*{Bases.}
 A base $\B$ is a pre-base that satisfies:
 \begin{enumerate}
	 \item
		For all $A\in\var$ and $R\in\dom(\B)$
		we have 
		$|\PD_R^\B(A)|\leq \norm{A}$. 
	 \item For each $(A,\alpha,R)\in\B_{dec}$ we have
		$\alpha=\PD_R^\B(\alpha)$.
\end{enumerate}

%
%
%
%

It is easy to verify the next fact:

\begin{proposition}
The size of any base is exponentially bounded, and the respective
conditions can be verified in exponential-time (w.r.t. the size of
$\calG$).
\end{proposition}	

\textbf{Base-generated equivalences.}
For a base $\B$ and any $R\in\dom(\B)$ 
we define the relation $\equiv_R^\B$ as follows:
\[
\alpha \equiv_R^\B \beta \defiff \PD_R^\B(\alpha) = \PD_R^\B(\beta).
\]
By definition, $A\equiv_R^\B\alpha$ for any triple $(A,\alpha,R)$ in
$\B$ (i.e., for any  $(A,\alpha,R)\in\B_{dec}$ and any
$(A,XA,R)\in\B_{prop}$; in the latter case we have $X\in\red^\B(A,R)$).

We put
\begin{center}
$\PD^\B(\alpha)=\PD^\B_\emptyset(\alpha)$
and $\equiv^\B=\equiv^\B_\emptyset$.
\end{center}

We also generalize~(\ref{eq:propag}),
defining $\red_R^\B(\alpha)$, rather written 
as $\red^\B(\alpha,R)$,
by the following inductive definition:
\begin{itemize}
  \item $\red^\B(\eps,R) = R$ (in particular, $\red^\B(\eps,\emptyset)
	  =\emptyset$);
 \item if $B$ is a $(\B, R)$-prime, then
	 $\red^\B(\alpha B,R) =  \red^\B(\alpha, R')$
where $R'=\red^\B(B,R)$;
 \item $\red^\B(\alpha,R) = \red^\B(\PD_R^\B(\alpha),R)$.
\end{itemize}

\smallskip

\textbf{Intended base.}
 Given an nBPA system $\calG=(\var,\act,\rules)$, an \emph{intended
 base} $\B$ arises by choosing the $R$-primes as described before
 Lemma~\ref{lem:rel-udp}, putting $(A,\PD_R(\alpha),R)$ in
 $\B_{dec}$ for each suffix-generated $R$ and each 
 non-$R$-prime $A$, and putting $(A,XA,R)$ in $\B_{prop}$ iff
$A$ is an $R$-prime and 
 $X\in\red(A,R)$.

 Corollary~\ref{cor:realprimedec}(1) shows that for any intended base $\B$
we have   $\equiv^\B_R=\sim_R$; in particular,
$\equiv^\B=\sim$, i.e., $\equiv^\B$ is the maximal branching
 bisimulation.

 Now we look for some suitable ``syntactic'' conditions guaranteeing
 that $\equiv^\B$, for a given (not necessarily intended) base $\B$, is a
branching bisimulation.

Each $\equiv^\B$ is symmetric, and we need to guarantee that for
each $\alpha\equiv^\B\beta$ and each transition
$\alpha\trans{a}\alpha'$ 
we have that
either $a=\tau$ and  
$\alpha'\equiv^\B\beta$, or 
there is a corresponding response 
$\beta=\beta_0\trans{\tau}\beta_1\trans{\tau}\cdots\trans{\tau}\beta_k\trans{a}\beta'$
where $\alpha'\equiv^\B\beta'$
and $\alpha\equiv^\B\beta_i$ for all $i\in\{1,2,\dots,k\}$.

\smallskip

\textbf{Legal-move outcomes, consistent bases.}
A natural idea is to define \emph{the set of $(\B,R)$-legal (move) outcomes
for}
$\alpha\in\var^*$, denoted $\LM^\B_R(\alpha)$
as a subset of $\act\times\var^*$:
\begin{quote}
$(a,\alpha')\in\LM^\B_R(\alpha)$ if 
\smallskip
\begin{enumerate}[i)]
	\item
		either $a=\tau$ and
$\alpha'=\PD^\B_R(\alpha)$,
\item
or there is a sequence 
$\alpha=\alpha_0\trans{\tau}\alpha_1\trans{\tau}\cdots\trans{\tau}\alpha_k\trans{a}\alpha''$
where $\alpha'=\PD^\B_R(\alpha'')$ and
$\PD^\B_R(\alpha_i)=\PD^\B_R(\alpha)$ for all $i\in\{1,2,\dots,k\}$.
\end{enumerate}
\end{quote}
A \emph{pair} $(\alpha,\beta)$ is \emph{$(\B,R)$-consistent} 
if $\alpha\equiv^\B_R\beta$
and
$\LM^\B_R(\alpha)=\LM^\B_R(\beta)$.

A \emph{base} $\B$ is \emph{consistent} if for each $(A,\alpha,R)$ in $\B$ we have
that $(A,\alpha)$ is a $(\B, R)$-consistent pair. 

\smallskip

\textbf{Rest of a proof of Theorem~\ref{th:brbisnexptime}.}
We now show that checking consistency of a base can be done in
exponential-time w.r.t. the size of $\calG$, that any intended base is
consistent, and that 
$\equiv^\B$ is a branching bisimulation for any consistent base $\B$.
After establishing these facts, a desired algorithm is obvious.

\begin{lemma}\label{lem:baseexptime}
Checking consistency of a base can be done 
in exponential time (w.r.t. the size of $\calG$).
\end{lemma}

\begin{proof}
Given a base $\B$ (whose size is exponentially bounded w.r.t. $\calG$
by definition), we need to verify that $\LM^\B_R(A)=\LM^B_R(\alpha)$ for each triple
$(A,\alpha,R)$ in $\B$.
There are at most exponentially many triples, 
so it suffices to focus on one of them;
we have either  $(A,\alpha,R)\in\B_{dec}$, in which case
$\alpha=\PD^\B_R(A)$,
or $(A,\alpha,R)=(A,XA,R)\in\B_{prop}$; 
in both cases we thus have $\PD^\B_R(A)=\PD^\B_{R}(\alpha)$,
i.e., $A\equiv^\B_R\alpha$.

It suffices to show how to compute 
$\LM^\B_{R}(\beta)$ (for any $\beta$). We first define the
\emph{$\tau$-closure} $\TC^\B_R(\beta)$ as the least set $\calT$ that
contains $\beta$ and satisfies the following condition:

for any $Y\gamma\in \calT$, $R'=\red^\B_R(\gamma,R)$, and a rule
$Y\trans{\tau}\delta$ of $\calG$, we have: 

\begin{enumerate}[i)]
	\item
	if $Y\in R'$ and
$\delta\in (R')^*$, then 
$\delta\gamma\in \calT$;
\item
	if $Y\not\in R'$ and
$\delta=\delta_1 Z\delta_2$ where $\delta_2\in (R')^*$, 
$Z\not\in R'$, and $\PD^\B_R(\delta_1 Z)=\PD^\B_R(Y)$,
then 
$\delta_1 Z\gamma\in \calT$.
\end{enumerate}		
We can easily check that
$\LM^\B_R(\beta)=\{(\tau,\PD^\B_R(\beta))\}\cup\{(a,\PD^\B_R(\mu))\mid
\gamma\trans{a}\mu$ for some $\gamma\in \TC^\B_R(\beta)\}$.
A dynamic-programming algorithm computing $\LM^\B_R(\beta)$ is thus
obvious; 
its time-complexity is bounded by the length of $\beta$ multiplied by
an exponential function of the size of $\calG$.
\qed
\end{proof}


\begin{lemma}\label{lem:intbasecons}
Any intended base is consistent.
\end{lemma}
\begin{proof}
Suppose $\B$ is an intended base.
By the definition of bases,
any $(A,\alpha,R)$ in $\B$ satisfies 
$A\equiv^\B_R\alpha$; moreover, we have
$PD^\B_R(A)=PD^\B_R(\alpha)\neq\eps$. We also note that 
$(\tau,PD^\B_R(A))$ belongs to both 
$\LM^\B_R(A)$ and $\LM^\B_R(\alpha)$.

Let $(a,\beta')\in\LM^\B_R(A)$, 
where $a\neq\tau$ or $\beta'\neq \PD^\B_R(A)$,
due to the sequence
\begin{center}
$A=\beta_0\trans{\tau}\beta_1\trans{\tau}\cdots\trans{\tau}\beta_k\trans{a}\beta''$
\end{center}
where 
$\PD^\B_R(\beta_i)=\PD^\B_R(A)$ for all $i\in\{0,1,\dots,k\}$
and $\beta'=\PD^\B_R(\beta'')$.
Since $\B$ is an intended base
(and thus $\equiv^\B_R=\sim_R$), we have
$A\sim_R\alpha\sim_R\beta_0\sim_R\beta_1\sim_R\cdots\sim_R\beta_k$;
we also have either $a\neq\tau$ or $\beta''\not\sim_R A$.
There must be a response from $\alpha$ (composed from responses to
the transitions $\beta_{i-1}\trans{\tau}\beta_i$ and 
$\beta_{k}\trans{a}\beta''$) of the form
\begin{center}
$\alpha=\alpha_0\trans{\tau}\alpha_1\trans{\tau}\cdots\trans{\tau}\alpha_\ell\trans{a}\alpha''$
\end{center}
where $\alpha_i\sim_R A$ for all $i\in\{0,1,\dots,\ell\}$, and 
$\alpha''\sim_R\beta''$. Hence $PD^\B_R(\alpha_i)=PD^\B_R(A)$
(for all $i\in\{0,1,\dots,\ell\}$)
and 
 $PD^\B_R(\alpha'')=PD^\B_R(\beta'')=\beta'$. Therefore 
 $(a,\beta')\in\LM^\B_R(\alpha)$, and thus  
$\LM^\B_R(A)\subseteq \LM^\B_R(\alpha)$.

Analogously we derive
$\LM^\B_R(\alpha)\subseteq \LM^\B_R(A)$, and thus 
$\LM^\B_R(A)=\LM^\B_R(\alpha)$.
\qed
\end{proof}

The remaining fact is captured by 
Lemma~\ref{lem:reductiontobase};
its main technical point is shown by the next proposition:

\begin{proposition}\label{prop:consLMsame}
If $\B$ is a consistent base
and $\alpha,\beta\in\var^*$ satisfy 
$\PD^\B_R(\alpha) = \PD^\B_R(\beta) \neq \eps$, then  
$\LM^\B_R(\alpha)=\LM^\B_R(\beta)$.
\end{proposition}


\begin{proof}
Suppose $\B$ is a consistent base and 
$\PD^\B_R(\alpha)=\PD^\B_R(\beta)\neq\eps$.
The definition of $\PD^\B_R(..)$ implies that
 $\alpha$ can be ``transformed'' into $\beta$ by using 
triples $(A,\gamma,R')$ from $\B$; we first transform $\alpha$ into 
$\PD^\B_R(\alpha)$, and then, by a ``backward-transformation'' we
change $\PD^\B_R(\alpha)=\PD^\B_R(\beta)$ into $\beta$. 
There is thus a (transformation) sequence 
\[
\alpha = \gamma_0, \gamma_1, \ldots, \gamma_m = \beta
\]
where $\PD^\B_R(\gamma_i)=\PD^\B_R(\alpha)$ for all $i\in\{0,1,\dots,m\}$
and each pair $(\gamma_i,\gamma_{i+1})$ satisfies the following:
\begin{center}
one of $\gamma_i,\gamma_{i+1}$ is in the form 
$\nu X \mu$ and the other in the form $\nu \delta \mu$, 
\end{center}
where 
$\mu=\PD^\B_R(\mu)$ and for $R'=\red^\B(\mu,R)$ we have 
\begin{enumerate}[i)]
	\item
		either $X\not\in R'$ and
		$(X,\delta,R')\in\B_{dec}$,
	\item
		or 
$X\in R'$ and $\delta=\eps$.
\end{enumerate}
It thus suffices  to check that we have  
$\LM^\B_R(\gamma_i) = \LM^\B_R(\gamma_{i+1})$ for each
$i\in\{0,1,\dots,m{-}1\}$. For a given pair $(\gamma_i,\gamma_{i+1})$
in the above form we now aim to show
that 
\begin{center}
$\LM^\B_R(\nu X \mu) = \LM^\B_R(\nu \delta \mu)$.
\end{center}
By our definitions, $(\tau,\PD^\B_R(\nu X \mu))$ is in both 
$\LM^\B_R(\nu X \mu)$ and $\LM^\B_R(\nu \delta \mu)$.
Suppose now that $(a,\rho)\in\LM^\B_R(\nu X \mu)$, where 
$a\neq\tau$ or $\rho\neq \PD^\B_R(\nu X \mu)$, due to a sequence 
\begin{equation}\label{eq:onelegmove}
\nu X
	\mu=\rho_{0}\trans{\tau}\rho_1\trans{\tau}\cdots\trans{\tau}\rho_k\trans{a}\rho'
\end{equation}	
where $\PD^\B_R(\rho_i)=\PD^\B_R(\nu X\mu)$ (for all
$i\in\{0,1,\dots,k\}$) and $\PD^\B_R(\rho')=\rho$.
If~(\ref{eq:onelegmove}) ``does not use $X$'', i.e., can be written as 
\begin{equation}\label{eq:oneform}
\nu X
\mu=\nu_0X\mu\trans{\tau}\nu_1X\mu\trans{\tau}\cdots\trans{\tau}\nu_k
X\mu\trans{a}\nu'X\mu
\end{equation}
(where $\nu=\nu_0\trans{\tau}\nu_1\trans{\tau}\cdots\trans{\tau}\nu_k
\trans{a}\nu'$), then we obviously have 
$(a,\rho)\in\LM^\B_R(\nu \delta \mu)$ (as shown by~(\ref{eq:oneform})
when $X$ is replaced with $\delta$).

We recall that $R'=\red^\B(\mu,R)$, and 
put $R''=\red^\B(X,R')=\red^\B(\delta,R')$.

If $\nu\not\in (R'')^*$, 
then~(\ref{eq:onelegmove}) must be always of the form 
~(\ref{eq:oneform}), by definition of legal outcomes; we thus further
assume  $\nu\in (R'')^*$. If $X\not\in R'$ (and thus $\delta\not\in
(R')^*$), then~(\ref{eq:onelegmove}) is 
either of the form~(\ref{eq:oneform})
or it uses $X$ but not $\mu$; in the latter case
we have a sequence $\nu X
\mu=\nu_0X\mu\trans{\tau}\nu_1X\mu\trans{\tau}\cdots\trans{\tau}\nu_k
X\mu\cdots\trans{a}\bar{\rho}\mu$ where $\nu_k=\eps$
and
$(a,\rho)=(a,\bar{\rho}\mu)$ where $(a,\bar{\rho})$ is in 
$\LM^\B_{R'}(X)$, and thus also in $\LM^\B_{R'}(\delta)$ since
$(X,\delta,R')\in\B_{dec}$ and
$\B$ is consistent; hence
$(a,\bar{\rho}\mu)=(a,\rho)\in\LM^\B_R(\nu\delta\mu)$. 

If ($\nu\in (R'')^*$ and) $X\in R'$, then 
$\delta=\eps$ and
$\mu\neq\eps$
(since $\PD^\B_R(\alpha)\neq\eps$); let us write
$\mu=A\mu'$ where $A$ is a $(\B,R''')$-prime for 
$R'''=\red(\mu',R)$ and we have $(A,XA,R''')\in\B_{prop}$. 
Now~(\ref{eq:onelegmove}) can also ``use $A$'', but cannot ``use $\mu'$'',
in which case 
$(a,\rho)=(a,\bar{\rho}\,\mu')$ where $(a,\bar{\rho})$ is in 
$\LM^\B_{R'''}(XA)$. Due to consistency of $\B$ we have
$(a,\bar{\rho})\in \LM^\B_{R'''}(A)$, and thus 
$(a,\bar{\rho}\mu')=(a,\rho)\in \LM^\B_{R}(\nu\delta \mu)$.

Hence $\LM^\B_R(\nu X \mu) \subseteq \LM^\B_R(\nu \delta \mu)$; the
direction $\LM^\B_R(\nu \delta \mu)\subseteq \LM^\B_R(\nu X \mu)$
follows analogously.
\qed
\end{proof}

\begin{lemma}\label{lem:reductiontobase}
If $\B$ is a consistent base, then $\equiv^\B$ is a branching
bisimulation.
\end{lemma}

\begin{proof}
Suppose $\B$ is a consistent base, 
but $\equiv_\B$ is not a branching bisimulation.
Then there is a pair $(\alpha, \beta)$, where $\PD^\B(\alpha) =
\PD^\B(\beta)$,
and a transition $\alpha\trans{a}\alpha'$ that
has no adequate response (w.r.t. $\equiv^\B$).
This implies that $(a,\PD^\B_\emptyset(\alpha'))\in \LM^\B_\emptyset(\alpha)$
but $(a,\PD^\B_\emptyset(\alpha'))\not\in\LM^\B_\emptyset(\beta)$. 
Since 
$\PD^\B_\emptyset(\alpha) =
\PD^\B_\emptyset(\beta)\neq \eps$
(otherwise $\alpha=\eps$ and we have no transition
$\alpha\trans{a}\alpha'$), we get a contradiction with
Prop.~\ref{prop:consLMsame}.
\qed
\end{proof}

\section{Nondeterministic exponential-time algorithms}\label{sec:algor}

\subsubsection{Bisimilarity problem (Theorem~\ref{th:brbisnexptime}).}
As already mentioned, a proof of Theorem~\ref{th:brbisnexptime} is now
obvious.
A nondeterministic exponential-time algorithm, when given a normed
$\calG=(\var,\act,\rules)$ and $A,B\in\var$,  guesses an (at most exponential) 
representation of a pre-base $\B$, checks that it is a consistent base, and
that $\PD^{\B}(A)=\PD^\B(B)$. 
Lemmas~\ref{lem:baseexptime},~\ref{lem:intbasecons}, 
and~\ref{lem:reductiontobase} show that this nondeterministic
algorithm indeed works in exponential time and that it has a successful run
if, and only if, $A\sim B$ in $\calL_\calG$.

\subsubsection{Regularity problem (Theorem~\ref{th:regulnexptime}).}
By using the results of previous sections
we derive Lemma~\ref{lem:cruxregul}, which
is the crux of a nondeterministic exponential-time algorithm deciding
semantic finiteness of a given nBPA process.
But we first recall a general notion of the bisimilarity quotient
and note a simple fact.

\textbf{Brbis-quotient.}
Given an LTS $\calL=(\calS,\act,(\trans{a})_{a\in\act})$, 
the \emph{quotient-LTS} (w.r.t. branching bisimilarity)
is defined as
$\calL_\sim=(\{[s]; s\in\calS\}, \act,(\trans{a})_{a\in\act})$
where the states are the equivalence classes, 
hence $[s]=\{s'\mid s'\sim s\}$, and $[s_1]\trans{a}[s_2]$ iff there
are $s'_1,s'_2$ such that $s_1\sim s'_1$, $s_2\sim s'_2$, 
and $s_1\trans{a}s'_2$. 

Since $\{(s,[s])\mid s\in\calS\} \cup 
\{([s],s)\mid s\in\calS\}$ can be easily verified to be a branching
bisimulation (on the disjoint union of $\calL$ and $\calL_\sim$), we
have $s\sim [s]$.

\textbf{Brbis-(in)finiteness.}
We say that a \emph{state} $r$ in an LTS is \emph{brbis-finite} if 
$r$ is  branching bisimilar with 
some state in a finite LTS; otherwise $r$ is  \emph{brbis-infinite}.
We observe the next simple fact:

\begin{proposition}
Given an LTS $\calL=(\calS,\act,(\trans{a})_{a\in\act})$,
a state $r\in\calS$ is brbis-finite
if, and only if, the set
$\gensimreach(r)=\{[s]; r \trans{w} s
\textnormal{ for some } w \in \act^*\}$
of equivalence classes reachable from $r$ is finite.
\end{proposition}

The next lemma is a consequence, when we recall that 
$\equiv^\B\subseteq\sim$ 
for any consistent base $\B$
(by Lemma~\ref{lem:reductiontobase}) and 
$\sim=\equiv^\B$
for any intended base 
(as follows from Corollary~\ref{cor:realprimedec}(1)), which is also
consistent (by Lemma~\ref{lem:intbasecons}).

\begin{lemma}\label{lem:cruxregul}
Given an nBPA system $\calG=(\var,\act,\rules)$,
a configuration $\alpha\in\var^*$ is 
brbis-finite
if, and only if,
there exists a consistent base $\B$ such that the set 
\[
\pdreach^\B(\alpha) = \{\PD^\B(\beta) \mid \alpha \trans{w} \beta
\textnormal{ for some } w \in \act^*\}
\]
of $(\B,\emptyset)$-prime-decompositions
of configurations reachable from $\alpha$
is finite.
\end{lemma}

\begin{proof}
If $\B$ is a consistent base, 
then $\PD^\B(\beta_1)=\PD^\B(\beta_2)$ implies $\beta_1\sim \beta_2$.
Hence if $\pdreach^\B(\alpha)$ is finite, then we can reach 
only finitely many equivalence classes from $\alpha$.

On the other hand, if
$\alpha$ is branching bisimilar with a state 
in a finite LTS, then there are 
only
finitely many equivalence classes reachable from $\alpha$.
For any intended (and consistent) base $\B$ we thus have that
$\pdreach^\B(\alpha)$ is finite.
\qed
\end{proof}

\textbf{Finishing a proof of Theorem~\ref{th:regulnexptime}.}
A nondeterministic algorithm 
deciding brbis-finiteness of a given $\alpha$ (for a given nBPA
system $\calG=(\var,\act,\rules)$),
can just 
guess an (at most exponential) representation of a pre-base $\B$,
check that it is a consistent base, and that the set $\pdreach^\B(\alpha)$ is finite.
Checking consistency can be done in exponential time by
Lemma~\ref{lem:baseexptime};
it thus 
remains to show how to test finiteness of $\pdreach^\B(\alpha)$.

Given a (consistent) base $\B$, we 
say that $(A,R)$ is a \emph{PD-loop}, where $A\in\var$ and
$R\in\dom(\B)$, if there are $w\in\act^*$, $\beta\in\var^*$
such that $A\trans{w}A\beta$,
$\PD^\B_R(\beta)\neq\eps$, and $\red^\B(\beta,R)=R$.

We now show that $\pdreach^\B(\alpha)$ is infinite if, and only if,
there exists $\delta = A \gamma$ such that $\delta$ is reachable from $\alpha$,
$R = \red^\B(\gamma)$ and $(A, R)$ is a PD-loop;
in this case we say that $\alpha$ \emph{reaches a PD-loop}.

For the ``if'' direction note that then 
\[
\alpha \trans{v}
A \gamma \trans{w}A\beta \gamma \trans{w}
A\beta\beta \gamma \trans{w}
A\beta\beta\beta \gamma \trans{w}\cdots
\]
for some $v \in \act^*$.
This path visits $\alpha_0=A \gamma $, $\alpha_1=A\beta \gamma $, $\cdots$,
$\alpha_i=A\beta^i \gamma$, $\cdots$ where
the set $\{\PD^\B_\emptyset(\alpha_i)\mid i\in\N\}$ is infinite.

For the other direction, if 
$\pdreach^\B(\alpha)$ is infinite, then there must be a path
$\alpha=\alpha_0\trans{a_1}\alpha_1\trans{a_2}\alpha_2\trans{a_3}\cdots$
where the set  $\{\PD^\B_\emptyset(\alpha_i)\mid i\in\N\}$ is infinite (which
follows from K\"onig's Lemma).
Then there obviously must be $i<j$ such that $\alpha_i=A\gamma$, $\alpha_j=A\beta\gamma$, 
and $A\trans{w}A\beta$ (for some $w\in\act^*$),
where for $R=\red^\B(\gamma,\emptyset)$ we have 
$\red^\B(\beta\gamma,\emptyset)=\red^\B(\beta,R)=R$ and 
$\PD^\B_R(\beta)\neq\eps$; hence $(A, R)$ is a PD-loop.

For finding PD-loops we can construct a directed graph as follows. We
take the pairs $(A,R)$ where $A\in\var$, $R\in\dom(\B)$ as the
vertices.
Now we put an arc from $(A,R)$ to $(B,R')$ iff there is a rule
$A\trans{a}\gamma_1 B\gamma_2$ of $\calG$ such that 
$R'=\red^\B(\gamma_2,R)$; moreover, if there is such a case with 
$\PD^\B_{R}(\gamma_2)\neq \eps$, then we ``colour'' the arc as ``blue''.
It is easy to check that $(A,R)$ is a PD-loop iff there is a cycle from
$(A,R)$ to $(A,R)$ in the graph that contains at least one blue arc.

Let us call a pair $(B,R')$ \emph{PD-infinite}
if there is a path in the graph that starts in  $(B,R')$
and ends in a PD-loop  $(A,R)$.
Now we observe that $\alpha$ reaches a PD-loop iff it can be written 
$\alpha=\beta_1 X\beta_2$ where $(X,R)$ is PD-infinite for 
$R=\red^\B(\beta_2,\emptyset)$.
An algorithm proving Theorem~\ref{th:regulnexptime} is thus clear.

\bibliographystyle{splncs03}
\bibliography{citat}

\end{document}